\documentclass[french]{article}
\usepackage[latin1]{inputenc}
\usepackage[T1]{fontenc}
\usepackage{lmodern}
\usepackage[a4paper]{geometry}
\usepackage{babel}
\usepackage{amsthm}
\usepackage{amsmath}
\usepackage{amsfonts}
\usepackage{textcomp}

\newtheorem{pro}{Proposition}[section]
\newtheorem{corol}{Corollaire}[section]
\newtheorem{deff}{Définition}[section]
\newtheorem{thm}{Théorème}[section]
\newtheorem{lem}{Lemme}[subsection]
\newcommand{\trace}{\textbf{Tr}_{\mathbf{\mathbb{F}_{2^{n}}}/\mathbf{\mathbb{F}_{2^{n/2}}}}}
 \newcommand{\tr} {\textbf{Tr}_{\mathbb{F}_{2^{n}}}}

\title{\Large{New construction APN quadratic functions}}

\author{
      Zahid \bsc{Mounir}%
      \thanks{Université de Paris 8}}
  \date{\today}    

\begin{document}

\maketitle  
\newpage
\begin{abstract}
Le but de cet exposé est de détailler l'article de M\up{r} \bsc{Carlet}. Au passage je ferais un rappel sur quelques résultats intéressants en théorie des corps finis, puis je donnerais des preuves (nouvelles) de quelques résultats connus, ensuite je généraliserais la construction d'une famille de fonction APN. La référence du résultat
précédera ce dernier, en cas d'absence de référence, la preuve sera de l'auteur.
\end{abstract}

\section{Corps finies}

Certains résultats ne seront pas prouvés nous renvoyons le lecteur curieux à \cite{LDL}. Certains résultats ne requièrent pas la finitude du corps, nous renvoyons le lecteur à cette référence  \cite{Gal}, dans la suite $\mathbf{K}$ désigne un corps commutatif quelconque pas forcément fini.
\begin{pro} \cite{Gal}\label{double}
\'{E}tant donné un corps $\mathbf{K}$.$P\in \mathbf{K}[X]$. Le polynôme $P$ n'a pas de facteur carrés si et seulement si $\gcd(P,P')=1$

\end{pro}
\begin{pro} \cite{Gal} \label{gcd}
 $\mathbf{K}$ corps, $n \in \mathbb{N}_{>1}$, $ s\in \mathbb{N\up{*}}$, dans $\mathbf{K}[X]$ on a:
$$ \gcd(X^{s}-1,X^{n}-1)=X^{\gcd(s,n)}-1  $$
\end{pro}
\begin{corol}\cite{Gal}
 $\mathbf{K}$ corps , $n \in \mathbb{N}_{>1}$, $ s\in \mathbb{N\up{*}}$:
$$ X^{s}-1|X^{n}-1 \Leftrightarrow s|n \;\mathrm{dans}\;\mathbb{N\up{*}}$$
\end{corol}
\begin{pro} \cite{LDL}
 $\mathbb{F}$  un corps fini t.q $\#\mathbb{F}=q \; \mathrm{alors} \;\;\forall x \in \mathbb{F}\;$ \begin{equation}\label{corps}  x^{q}=x.\end{equation}
et inversement les solution de l'equation (\ref{corps}) sont exactement les éléments de $\mathbb{F}$ .
\end{pro}
\begin{corol} \cite{LDL} \label{inter}
 $p$ un nombre premier  $s,n \in \mathbb{N\up{*}}$.  alors :
$$ \mathbb{F}_{p^{s}}\bigcap \mathbb{F}_{p^{n}}=\mathbb{F}_{p^{\gcd(s,n)}}$$
\end{corol}
\begin{corol}\label{coro1}
 $\mathbb{F}_{2^{n}}$ un corps fini $ \forall \beta \in \mathbb{F}_{2^{n}}$  $\forall a \in \mathbb{N}$:
 $${\beta}^{2^{a}}={\beta}^{2^{a\bmod n}}$$
 \begin{proof}
 Par division euclidienne de $a$ par $n$, $\exists(q,r) \in \mathbb{N}^{2}$ tel que : $a=qn+r \quad \mathrm{avec}\; 0\leq r<n.$
 Donc $n|nq \Rightarrow \mathbb{F}_{2^{n}}\subset \mathbb{F}_{2^{nq}}\Rightarrow \forall \beta \in \mathbb{F}_{2^{n}},\;{\beta}^{2^{qn}}=\beta$\\
 on conclut avec : ${\beta}^{2^{a}}={\beta}^{2^{qn+r}}=({\beta}^{2^{qn}})^{2^{r}}={\beta}^{2^{r}}={\beta}^{2^{a \bmod n}}$
 
 \end{proof}
\end{corol}
\begin{corol}
Soit $n$ un entier pair non nul.
\renewcommand{\labelenumi}{{\normalfont (\roman{enumi})}}
 \begin{enumerate}
 \item $\; \forall x \in \mathbf{F}_{2^{n}},\; x^{2^{\frac{n}{2}}+1} \in \mathbf{F}_{2^{n/2}}$
 \item Si $n/2$ impair, alors $\mathbb{F}_{2^{n/2}}\cap \mathbb{F}_{2^{2}}=\mathbb{F}_{2}$.
  
 \end{enumerate}
 \end{corol}
 \begin{proof} 
 \indent  
 \renewcommand{\labelenumi}{{\normalfont (\roman{enumi})}}
 \begin{enumerate}
 \item Il suffit d'appliquer le corollaire \ref{inter}.
 \item En effet ${(x^{2^{\frac{n}{2}}+1})}^{2^{\frac{n}{2}}}=x^{2^{n}+2^{\frac{n}{2}}}=x^{2^{\frac{n}{2}}+1}$
 \end{enumerate}
  \end{proof}
  
  \begin{pro}\cite{Gal}
$p$ premier, $n\in \mathbb{N\up{*}}$ et $q=p^{n}$. Les $\mathbb{F}_p$-sous espaces vectorielles de $\mathbb{F}_{q} $ sont au nombre de :
$$\sum_{s=0}^{n} \frac{(p^{n}-1)(p^{n-1}-1)\ldots (p^{n-s+1}-1)}{(p^{s}-1)(p^{s-1}-1)\ldots (p-1)}$$
\end{pro}
\begin{proof}
 Pour $s \in \{1,\ldots,n\}$ , dénombrons les $\mathbb{F}_p$-sous espaces vectorielles de dimension $s$ de $\mathbb{F}_{q} $ .
 \begin{itemize}
 \item[\Huge{\textbf{.}}] Le premier vecteur étant choisi non nul : $p^{n}-1$ possibilités.
 \item[\Huge{\textbf{.}}] Le second vecteur,non  colinéaire au premiers : $p^{n}-p$ possibilités.
 \item[\Huge{\textbf{.}}] Le troisième vecteur non lié aux deux premiers : $p^{n}-p^{2}$ possibilités.
 \item[\Huge{\textbf{.}}]$\ldots$
 \item[\Huge{\textbf{.}}] Le $s$-ième vecteur non lié aux précédents : $p^{n}-p^{s-1}$ possibilités.
 
 \end{itemize}
 Il y'a donc $(p^{n}-1)(p^{n}-p)\ldots (p^{n}-p^{s-1})$ systèmes libres à $s$ éléments. Le même raisonnement montre qu'un $\mathbb{F}_p$-sous espace vectorielle de dimension $s$ de $\mathbb{F}_{q} $ admet:
 $(p^{s}-1)(p^{s}-p)\ldots (p^{s}-p^{s-1})$ bases. Le nombre de $\mathbb{F}_{p}$-espace vectorielle de dimension de dim $s$ est donc:
 $$\frac{(p^{n}-1)(p^{n}-p)\ldots (p^{n}-p^{s-1})}{(p^{s}-1)(p^{s}-p)\ldots (p^{s}-p^{s-1}}= \frac{(p^{n}-1)(p^{n-1}-1)\ldots(p^{n-s+1}-1)}{(p^{s}-1)(p^{s-1}-1)\ldots(p-1)}$$

\end{proof}

\subsection{Critère d'irréductibilité}\label{irréd}
\begin{pro}\cite{Gal} \label{deg3}
Soit $P\in \mathbf{K}[X]\; \mathrm{tel\, que}\; D\up{°}(P)\leq 3$.\par
$P $ est irréductible sur $\mathbf{K}$ si et seulement si $P$ n'a pas de racine dans $ \mathbf{K}$
\end{pro}
\begin{pro}\cite{Gal}
Soit $P\in \mathbf{K[X]} \; \text{avec}\; D\up{°}(P)=n$.\par
$P$ est irréductible si et seulement si $P$ n'a pas de racines dans toutes extension $\mathbf{L}/\mathbf{K}$
tel que: $[\mathbf{L}:\mathbf{K}]\leq n/2$.
\end{pro}
\begin{proof}
\indent
\begin{proof}[Condition nécéssaire]
$P$ irréductible sur $ \mathbf{K}$. Soit $\alpha \in \mathbf{L}$, racine de $P$ alors $\mathbf{K(\alpha)}$ est un corps de rupture de $P$. \\
Donc $[\mathbf{K(\alpha)}:\mathbf{K}]=n \Rightarrow [\mathbf{L}:\mathbf{K}]\geq n > n/2$.

\renewcommand{\qedsymbol}{}
\end{proof}
\begin{proof}[Condition suffisante]
Par Contraposition. Si $P$ n'est pas irréductible, il existe $(Q,R)\in {\mathbf{K[X]}}^{2}$.
tel que : $P=QR \;\mathrm{et} \;1\leq D\up{°}(R), D\up{°}(Q)< n$, sans perte de généralité on peut supposer que : $D\up{°}(Q)\leq \frac{n}{2}$. 
Soit $f$ un facteur irréductible de $Q$, et $\mathbf{L}=\mathbf{K(\alpha)}$ un corps de rupture de $f$ alors $\alpha \in \mathbf{L}$ est une
racine de $P(X)$ et $[\mathbf{L}:\mathbf{K}]= D\up{°}(f)\leq \frac{n}{2}$.
\renewcommand{\qedsymbol}{}
\end{proof}
\end{proof}
\begin{pro} \cite{Gal} \label{irré}
Soit $P\in \mathbf{K[X]}$ irréductible, $D\up{°}(P)=n$ et $\mathbf{L}/\mathbf{K}$ extension de degré $m$ de  $\mathbf{K}$ avec $\gcd(m,n)=1$
alors $P$ est irréductible dans $\mathbf{L[X]}$.

\end{pro}
\begin{proof}
Supposons  $P$ est réductible dans $\mathbf{L[X]}$, soit $f$ un facteur irréductible de $P \;\mathrm{dans}\; \mathbf{L[X]}$
alors $0<D\up{°}(f)<n$. Soit $M=\mathbf{L(\alpha)}$ un corps de rupture de $f$.\\
$P$ étant irréductible dans $\mathbf{K[X]}$, donc $\mathbf{K(\alpha)}$ corps de rupture de $P$ sur $\mathbf{K}$. Alors $[\mathbf{K(\alpha)}:\mathbf{K}]=n$ , donc $[\mathbf{M}:\mathbf{K}]=[\mathbf{M}:\mathbf{K(\alpha)}][\mathbf{K(\alpha)}:\mathbf{K}]$ 
est divisible par $n$. Or $[\mathbf{M}:\mathbf{K}]=[\mathbf{M}:\mathbf{L}].[\mathbf{L}:\mathbf{K}]=D\up{°}(f)\times m$.
Comme  $\gcd(m,n)=1$, il vient $n$ divise $D\up{°}(f)$, contradiction. 
\end{proof}
\textbf{Remarque:} La proposition \ref{irré},  peut être déduite de la proposition qui va suivre, si nous avons éviter, c'est pour insister sur son caractère générique.
\begin{pro}\cite{LDL}
Soit $P \in \mathbb{F}_{q}[X]$ irréductible de degré $n$ et soit $k\in \mathbb{N}^{*}$, alors $P$ se factorise en $d$ polynômes irréductibles sur $\mathbb{F}_{q^{k}}[X]$ de
degré $n/d$ avec $d=\gcd(n,k)$ 
\end{pro}
\subsection{Trace sur un corps }
\begin{deff}
Soient $\mathbf{K}=\mathbb{F}_{q}$ et $\mathbf{F}=\mathbb{F}_{q^{m}}$.
Pour tout $\alpha \in \mathbf{F}$, la trace $\textbf{Tr}_{\mathbf{F}/\mathbf{K}}(\alpha )$ de $\alpha$ sur $\mathbf{K}$ est définie par:
$$\textbf{Tr}_{\mathbf{F}/\mathbf{K}}(\alpha )=\alpha +\alpha^{q}+\ldots+\alpha^{q^{m-1}}$$
Si $\mathbf{K}$ est un corps premier la trace est dite absolue et on la note seulement $\textbf{Tr}_{\mathbf{F}}.$
\end{deff}

La trace a des propriétés intéressantes que nous énoncerons sous forme d'un théorème:
\begin{thm}\cite{LDL}
Soient $\mathbf{F}=\mathbb{F}_{q^{m}}$ et $\mathbf{K}=\mathbb{F}_{q}$ alors:
 \renewcommand{\labelenumi}{{\normalfont (\roman{enumi})}}
 \begin{enumerate}
 \item $\textbf{Tr}_{\mathbf{F}/\mathbf{K}}$ est une forme $\mathbf{K}$-linéaire non nulle surjective.
 \item Pour tout $a\in \mathbf{K}\;,\;\textbf{Tr}_{\mathbf{F}/\mathbf{K}}(a)=ma$.
 \item Pour tout $\alpha \in \mathbf{F},\; \textbf{Tr}_{\mathbf{F}/\mathbf{K}}(\alpha^{q})=\textbf{Tr}_{\mathbf{F}/\mathbf{K}}(\alpha)$ (La trace est stable par le \textbf{\itshape{Frobenius})}.
 \end{enumerate}
\end{thm}
\begin{pro}[Transitivité de la trace] \cite{LDL}\label{transi}
      
 Soient $\mathbf{K}$ un corps fini, $\mathbf{F}$ une extension finie de $\mathbf{K}$ et $\mathbf{E}$ une extension finie de $\mathbf{F}$. Alors
 
 $$ \textbf{Tr}_{\mathbf{E}/\mathbf{K}}=\textbf{Tr}_{\mathbf{F}/\mathbf{K}}o\textbf{Tr}_{\mathbf{E}/\mathbf{F}}$$

\end{pro}
\begin{corol}\label{coro}
Soit $n$ un entier pair non nul. Alors:
\begin{enumerate}
\item
$ \forall x \in \mathbb{F}_{2^{n/2}},\; \textbf{Tr}_{\mathbb{F}_{2^{n}}}(x)=0.\,(\;\textit{i.e}\;\mathbb{F}_{2^{n/2}}\subset \mathrm{\mathcal{K}er}(\textbf{Tr}_{\mathbb{F}_{2^n}}))$
\item Il existe $\omega \in \mathbb{F}_{2^{n}} \setminus \mathbb{F}_{2^{n/2}} \; \text{tel que}\; \trace(\omega)=1$
\end{enumerate}
\end{corol}
 
\begin{proof}
\indent
\begin{enumerate}
\item
Par définition : Pour tout $x \in \mathbb{F}_{2^{n}}\;,\; \textbf{Tr}_{\mathbf{\mathbb{F}_{2^{n}}}/\mathbf{\mathbb{F}_{2^{n/2}}}}(x)=x+x^{2^{\frac{n}{2}}}$
et donc $ \textbf{Tr}_{\mathbf{\mathbb{F}_{2^{n}}}/\mathbf{\mathbb{F}_{2^{n/2}}}}$ est nulle sur $\mathbb{F}_{2^{n/2}}$
on conclut avec $\textbf{Tr}_{\mathbb{F}_{2^{n}}}=\textbf{Tr}_{\mathbb{F}_{2^{n/2}}}o\textbf{Tr}_{\mathbf{\mathbb{F}_{2^{n}}}/\mathbf{\mathbb{F}_{2^{n/2 }}}}$( cf. proposition \ref{transi})
\item $\trace{}$ est une forme $\mathbb{F}_{2^{n/2}}$-linéaire non nulle.Donc,il existe $x_{0} \in \mathbb{F}_{2^{n}}$ tel que $\trace(x_{0})\neq 0$. Il suffit de choisir
$\omega=\frac{x_{0}}{\trace(x_{0})}$ est conclure par la $\mathbb{F}_{2^{n/2}}$ -linéarité de la  $\trace{}$.
\end{enumerate}
\end{proof}
\textbf{Remarque:} le point $(2)$ peut-être directement prouvé on utilisant la surjectivité de la trace, seulement je voulais donner une construction effective.

\begin{pro}\label{clef}
 Soient $n$ un entier pair non nul, $q=2^{n/2}$ et $c\in \mathbb{F}_{2^{n}}$ vérifiant : $c^{q+1}=1$ alors 
 $(\frac{1}{c^{2^{n-1}}})^{q}=\frac{c}{c^{2^{n-1}}}$ de plus on a:
 \[
 \forall \omega \in \mathbb{F}_{2^{n}} \quad \frac{\omega+c{\omega}^{q}}{c^{2^{n-1}}} \in \mathbb{F}_{2^{n/2}}.
 \]

\end{pro}
\begin{proof}
 $ c^{q+1}=1\Leftrightarrow c=\frac{1}{c^{q}}$
$\Rightarrow $
\begin{align*}
\left(\frac{1}{c^{2^{n-1}}}\right)^{q} &=c^{2^{n-1}} \\
                            &=c^{2^{n}-2^{n-1}}\\
                            &=\frac{c^{2^{n}}}{c^{2^{n-1}}}\\
                            &=\frac{c}{c^{2^{n-1}}}
        \end{align*}
   
  Soit $\omega \in \mathbb{F}_{2^{n}}$.\\
  $$\frac{\omega+c{\omega}^{q}}{c^{2^{n-1}}}=\frac{w}{c^{2^{n-1}}}+{\omega}^{q} \frac{c}{c^{2^{n-1}}}=\frac{\omega}{c^{2^{n-1}}}+\left(\frac{\omega}{c^{2^{n-1}}}\right)^{q}\\
  = \trace\left(\frac{\omega}{c^{2^{n-1}}}\right) \in \mathbb{F}_{2^{n/2}}$$

\end{proof}
\subsection{Permutation particulière sur un corps fini}\label{per}
\begin{lem}\label{groupe}
Soit $G=<x>$ un groupe cyclique d'ordre $n$. Alors $ \forall d,\;d\mid n \;\mathrm{dans}\; \mathbb{N^*}$\\
il existe un unique sous groupe d'ordre $d$ de $G$. Il  est engendré par $x^{k}$, $k=\frac{n}{d}$

\end{lem}
\begin{pro}\label{Per}
Soient $\mathbf{F}_{q}$ un corps fini, $i\in \mathbb{N}^{*}$.
\renewcommand{\labelenumi}{{\normalfont (\roman{enumi})}}
\begin{enumerate}
\item $X^{i}$ permute $\mathbf{F}_{q}$ si et seulement si $\gcd(i,q-1)=1$.
\item Le nombre de i\ieme{} puissance non nulle dans $\mathbf{F}_{q} $(\,\textit{i.e}$ \;\#{{\mathbb{F^{*}}_{q}}^{i}}\,$) est $(q-1)/\gcd(i,q-1)$.
\end{enumerate}
\end{pro}
\begin{proof}

Remarquons d'abord que $0$ est la seule solution de l'équation $X^{i}=0$.\\
Considérons le morphisme de groupe  $\mathfrak{F_{i}}\colon x\in {\mathbb{F}_{q}}^{*} \to x^{i}\in \mathbb{F^{*}}_{q}$. 
Calculons son noyau:\\
\[
\mathrm{\mathcal{K}er}(\mathfrak{F_{i}})=\{\,x \in {\mathbb{F}_{q}}^{*} \mid x^{i}=1\,\}
\]
On a  $x\in \mathrm{\mathcal{K}er}(\mathfrak{F_{i}})\Rightarrow \;\mathrm{ord}(x)\mid i \Rightarrow \mathrm{ord}(x)\mid d\,,\,\mathrm{ou}\; d=\gcd(i,q-1)\Rightarrow x\in \mathrm{\mathcal{K}er}(\mathfrak{F}_{d}) $
l'inverse est évidente. Donc:
\[
\mathrm{\mathcal{K}er}(\mathfrak{F_{i}})=\mathrm{\mathcal{K}er}(\mathfrak{F}_{d}) \;\mathrm{avec}\; d=\\gcd(i,q-1).
\]
Le polynôme $P(x)=X^{d}-1$ n'a que des racines simples. ( Voir Proposition \ref{double} ). 
\'{E}tant un polynôme de degré $d$, $P$ a donc $d$ racine distinct.
 $d\mid (q-1)$ implique  $X^{d}-1\mid X^{(q-1)}-1$. Donc tous les racines de $P$ sont dans ${\mathbb{F}_{q}}^{*}$.\\
Ceci entraîne $ \#{\mathrm{\mathcal{K}er}(\mathfrak{F}_{d})}=d$. \\
  $ \mathfrak{F_{i}}$ est un isomorphisme si et seulement si $d=1$, c-à-d $\gcd(i,q-1)=1$.

D'après le 1\up{er} théorème d'isomorphisme 
\[
{\mathbb{F}_{q}}^{*}/\mathrm{\mathcal{K}er}(\mathfrak{F_{i}}) \cong \mathcal{I}m(\mathfrak{F_{i}})
\]
or $\mathcal{I}m(\mathfrak{F_{i}}):= $\;les i\ieme{} puissances non nulles dans $\mathbf{F}_{q}$, ce qui conclut la preuve.
\end{proof}
\noindent  \textbf{Remarque:}
\begin{enumerate}
\item $\mathcal{I}m(\mathfrak{F_{i}})=\mathcal{I}m(\mathfrak{F}_d)$; c-à-d les i\ieme{} puissances non nulles dans $\mathbf{F}_{q}$ sont exactement les d\ieme{} puissances non nulles dans $\mathbf{F}_{q}$.
\item $\forall (x,y)\in {\mathbf{F}_{q}}^{2} \quad  x^{i}=y^{i} \Leftrightarrow x^{d}=y^{d}$ ( Découle de $\mathrm{\mathcal{K}er}(\mathfrak{F_{i}})=\mathrm{\mathcal{K}er}(\mathfrak{F}_{d})$ ).
\item Soit $\alpha$ un générateur de ${\mathbb{F}_{q}}^{*}$ \\
$\mathrm{\mathcal{K}er}(\mathfrak{F}_{d})$ est un sous-groupe de ${\mathbb{F}_{q}}^{*}$ d'ordre $d$, il est donc engendré par $\xi={\alpha}^{(q-1)/d}$ ( cf. lemme \ref{groupe} ).
\[
\mathrm{\mathcal{K}er}(\mathfrak{F}_{d})=\left\{\xi^{k};k=0,\ldots,d-1\right\}
\]
\item Définissons une relation d'équivalence sur ${\mathbb{F}_{q}}^{*}$ par: \\
$\forall (x,y) \in {{\mathbb{F}_{q}}^{*}}^2  ,\quad x \mathfrak{R}y  \; \mathrm{si et seulement si}\; y\in x\mathrm{\mathcal{K}er}(\mathfrak{F_{i}}) $\\
C'est bien une relation d'équivalence et : $\forall x \in {\mathbb{F}_{q}}^* \quad \mathcal{C}l(x)=x\mathrm{\mathcal{K}er}(\mathfrak{F}_{d})=\left\{x\xi^{k};k=0,\ldots,d-1\right\}$. Les classes forment une partition de l'ensemble en question:
\[
{\mathbb{F}_{q}}^{*}=\bigcup_{x\in \mathcal{I}}\mathcal{C}l(x)  \quad \mathrm{avec} \;\#\mathcal{I}=\frac{q-1}{d} \;\mathrm{et}\;\#\mathcal{C}l(x)=d .
\]

\end{enumerate}

\subsection{$\mathbb{F}_{2^n} \; \mathrm {versus}\; \mathbb{F}_{2^{n/2}}\times \mathbb{F}_{2^{n/2}}$} 

Certaines fonctions sont définies sur $\mathbb{F}_{2^{n/2}}\times \mathbb{F}_{2^{n/2}}$, on aimerait bien expliciter leur représentation univariée et pour cela il faut 
les définir sur  $\mathbb{F}_{2^n}$, nous allons voir comment:\par
Soit   $\omega \in \mathbb{F}_{2^{n}} / \mathbb{F}_{2^{n/2}}$. $(1,\omega)$ est une base du $\mathbb{F}_{2^{n/2}}$-espace vectorielle $\mathbb{F}_{2^{n}}$.
Pour tout $X \in \mathbb{F}_{2^n}$, il existe un unique couple $(x,y)$ dans $\mathbb{F}_{2^{n/2}}$ tel que 
\begin{equation}\label{E:firstequ}
X=x+\omega y
\end{equation} 
( \textit{i.e}$ \;\mathbb{F}_{2^{n}}=\mathbb{F}_{2^{n/2}}\oplus \omega \mathbb{F}_{2^{n/2}} ).$ \par
Nous allons expliciter $x$ et $y$ en fonction de $X$. Appliquons $ \trace $ a l'equation~\eqref{E:firstequ}, on obtient $ \trace(X)=y \trace(\omega)$
de même d'après~\eqref{E:firstequ}: $X\omega^{2^{n/2}}+X^{2^{n/2}}\omega =x(\omega^{2^{n/2}}+\omega)$
c-à-d $x=\frac{\trace\left(X\omega^{2^{n/2}}\right)}{\trace(\omega)}$   \par
Nous avons le $\mathbb{F}_{2^{n/2}}$-isomorphisme d'espace vectorielle suivant:

$ X \in \mathbb{F}_{2^{n}} \rightarrow (x,y)\in \mathbb{F}_{2^{n/2}}\times \mathbb{F}_{2^{n/2}}$ \\
avec $ x=\frac{\trace\left(X\omega^{2^{n/2}}\right)}{\trace(\omega)} \; \mathrm{et}\; y=\frac{\trace(X)}{\trace(\omega)}$\\
\noindent  \textbf{Remarque:}
\begin{enumerate}
\item D'après le corollaire~\ref{coro} $\omega$ peut-être choisi tel que $\trace(\omega)=1$ ce qui simplifie considérablement le calcul.
\item Dans le cas ou $n/2$ est impair, on a davantage de simplification, il suffit de choisir $\omega$ élément primitif de $\mathbb{F}_{4}$ 
on a d'après la corollaire~\ref{coro1}: $\omega^{2^{n/2}}=w^{2}$
\item Rien n'empêche de prendre $\omega=\alpha^{2^{n/2}-1}$ et dans ce cas:$\trace(\omega)=\omega+\omega^{-1}$
\end{enumerate}

\section{Fonctions Courbes}
Nous allons donner quelques résultats intéressants, pour une étude plus approfondie, nous envoyons à \cite{Car}.
\subsection{Transformée de Walsh}
La Transformée de \textbf{Walsh} d'une fonction booléenne $f$ et la transformée de \textbf{Fourier} de sa fonction signe. Son expression est donc:\\
 $$\forall u \in \mathbb{F}_{2}^{n}\quad \hat{\chi}_{f}(u)=\sum_{x\in\mathbb{F}_{2}^{n}} (-1)^{f(x)+u \cdot x} \; \textrm{ où $u \cdot x $ désigne le produit scalaire dans $\mathbb{F}_{2}^{n}$  }$$
 Si $\hat{\chi}_{f}(0)=0$ alors  $f$ est équilibrée.\\
 Si $ \forall u \in \mathbb{F}_{2}^{n} \; \hat{\chi}_{f}(u)=\pm 2^{n/2}$ alors $f$ est dite \textbf{courbe}.

\subsection{\'{E}tude d'une fonction booléenne particulière}
Dans la suite nous identifions $\mathbb{F}_{2^{n}}$ a $\mathbb{F}_{2}^{n}$ et nous posons $u\cdot x=\textbf{Tr}_{\mathbb{F}_{2^{n}}}(xu)$\\
 \'{E}tude de la fonction $f(x)=\textbf{Tr}_{\mathbb{F}_{2^{n}}}(a x^{i}) \;\mathrm {avec}\; a\neq0$.
 On peu déjà remarquer que $\gcd(i,2^{n}-1)\neq1$ sinon $f$ serait équilibrée, une telle fonction n'est jamais courbe.
 \subsubsection{Cas $a\in {\mathbb{F}}^{*i}_{2^{n}}$}
 \begin{pro}
 $f$ est courbe si et seulement si $g(x)=\textbf{Tr}_{\mathbb{F}_{2^{n}}}(x^{i}) $ est courbe.
 \end{pro}
 \begin{proof}
 $ a \in \mathbb{F}^{*i}_{2^{n}} \Leftrightarrow \exists b\in {\mathbb{F}}^{*}_{2^{n}} \mid \, a=b^{i}.$
 Soit  $\beta \in \mathbb{F}_{2^{n}}\;$
 \begin{align*}
     \hat{\chi}_{f}(\beta)&=\sum_{x\in\mathbb{F}_{2}^{n}} (-1)^{\tr(ax^{i})+\tr(\beta x)}\\
                           &=\sum_{x\in\mathbb{F}_{2}^{n}} (-1)^{\tr((bx)^{i})+\tr(\beta x)}\\
                           &=\sum_{x\in\mathbb{F}_{2}^{n}} (-1)^{\tr(x^{i})+\tr(\frac{\beta}{b} x)}\\
                           &=\hat{\chi}_{g}\left(\frac{\beta}{b}\right)
 \end{align*}
 \end{proof}
\noindent  \textbf{Remarque:}
L'étude que j'ai menée sur le caractère courbe  de $g$ sur $\mathbb{F}_{2^{k}}\,\mathrm{où}\; k=4,\ldots,22$, a montré que $g$ est courbe uniquement sur
$\mathbb{F}_{2^{8}}$ avec l'exposant $i=(1+j)15,\;j=0,\ldots,15$. Ce qui correspond a un exposant de \textbf{Dillon}, étrangement $ \mathbb{F}_{2^{8}}$ c'est le corps où est définit l'\textbf{A.E.S}, y'a t-il une causalité??.
\subsubsection{Cas $a\notin {\mathbb{F}}^{*i}_{2^{n}}$}
Dans cette partie nous allons tirer profit de l'étude que nous avons réalisé dans la Proposition~\ref{Per}. Les notations sont celles de la-dite Proposition et de la remarque 
qui l'a suivie avec $q=2^{n}$.
Soit $\beta \in \mathbb{F}_{2^{n}} $\\ 
  $\xi={\alpha}^{(2^{n}-1)/d} ,\mathrm{ou}\; d=\gcd(2^{n}-1,i)$
 \begin{align*}
     \hat{\chi}_{f}(\beta \xi)&=\sum_{x\in\mathbb{F}_{2}^{n}} (-1)^{\tr(ax^{i})+\tr(\beta \xi x)}\\
                             &=\hat{\chi}_{f}(\beta) .\left(\,\mathrm{car}\, x\longmapsto x\xi \,\text{est une permutation sur}  \;\mathbb{F}_{2^{n}}\,\right)
  \end{align*}
  Pour résumer :
   \begin{equation*}
   \boxed{\forall y \in \mathcal{C}l(x)\quad\hat{\chi}_{f}(x)=\hat{\chi}_{f}(y) }
   \end{equation*}
   La transformée de \textbf{Walsh} est constante sur les classes,ceci peut conduire à un algorithme plus rapide(?). On peut aussi réécrire la \textbf{TW} autrement:\\
   
   \begin{align}
   \hat{\chi}_{f}(\beta)&=\sum_{x\in\mathbb{F}_{2}^{n}} (-1)^{\tr(ax^{i})+\tr(\beta x)}\\
                        &=1+\sum_{x\in\mathbb{F}^{*}_{2^{n}}} (-1)^{\tr(ax^{i})+\tr(\beta x)}\\
                        &=1+\sum_{x\in \mathcal{I}}(-1)^{\tr(ax^{i})}\sum_{y\in \mathcal{C}l(x) } (-1)^{\tr(\beta y)}\\
                        &=1+\sum_{x\in \mathcal{I}}(-1)^{\tr(ax^{i})} \sum_{k=0}^{d-1} (-1)^{\tr(\beta x \xi^{k})}\label{a}
   \end{align}
   
   \begin{pro}
   Si $f(x)=\textbf{Tr}_{\mathbb{F}_{2^{n}}}(a x^{i}) $ est \textbf{courbe}, alors:
   \begin{equation*}
   \hat{\chi}_{f}(0)=
   \begin{cases}
   2^{n/2}, &\text{ssi}\;\gcd(i,2^{n/2}+1)=1\\
   -2^{n/2}, &\text{ssi}\;\gcd(i,2^{n/2}-1)=1
   \end{cases}
   \end{equation*}
   
   \end{pro}
   \begin{proof}
   $d=\gcd(i,2^{n}-1)=\gcd(i,2^{n/2}-1).\gcd(i,2^{n/2}+1).$\\
    $k=\gcd(i,2^{n/2}-1)$ et $l=\gcd(i,2^{n/2}+1). $En posant $\beta=0$ dans l'égalité~\eqref{a}
      \begin{equation*}
      \hat{\chi}_{f}(0)-1=kl\sum_{x\in \mathcal{I}}(-1)^{\tr(ax^{i})}
      \end{equation*}
      
   \begin{description}
   \item[1\up{er} cas:] $\hat{\chi}_{f}(0)=2^{n/2}$ ceci entraîne $2^{n/2}-1=kl\sum_{x\in \mathcal{I}}(-1)^{\tr(ax^{i})}$\\
   Alors $l|2^{n/2}-1 \;\mathrm{comme}\;\gcd(l,2^{n/2}-1)=1$, cela implique  $l=1$

   \item[2\up{me} cas:] $\hat{\chi}_{f}(0)=-2^{n/2}$ le même raisonnement  conduit à $k=1$.
   \end{description}
   \end{proof}
 \section{Construction d'une classe APN a partir d'une fonction Bent}
  \begin{deff}
  Une $(n,n)-$fonction $F$ est dite \textbf{APN} si:\\
   $\forall a \in \mathbb{F^{*}}_{2^{n}} , \forall b \in \mathbb{F}_{2^{n}}$. L'équation : $F(x)+F(x+a)=b $ à au plus $0$ ou $2$ solutions .
  \end{deff}
  Notation: 
  \begin{equation}
  \forall a \in \mathbb{F^{*}}_{2^{n}}, D_{a}F(x)=F(x+a)+F(x). 
  \end{equation}
  \renewcommand{\labelenumi}{{\normalfont (\roman{enumi})}}
  \begin{pro}\cite{Car}
  \indent
  \begin{enumerate}
  \item  $B$ est courbe \textbf{si et seulement si} $D_{a}B$ est équilibré.
  \item $B$ quadratique  \textbf{alors} $D_{a}B$ est affine.
   \end{enumerate}
  \end{pro}
  \begin{proof}
  voir livre \cite{Car}.
  \end{proof}
  \noindent  \textbf{Remarque:}
  Si $B$ une $(n,n/2)-$courbe quadratique, alors les solutions de $D_{a}B(x)=b$ avec $(a,b)\in \mathbb{F^{*}}_{2^{n}}\times \mathbb{F}_{2^{n/2}}$ est un sous-espaces affine de     $\mathbb{F}_{2^{n}}$ affinement isomorphe a  $\mathbb{F}_{2^{n/2}}$. Seulement c'est isomorphisme il n'est pas simple de l'explicité, d'autant plus qu'il depend de $a$ et $b$. On va voir que dans le cas de la fonction simple de \textbf{Mairona Mac Farland} ce n'est pas le cas. L' auteur de l'article \cite{art} a exploité cette idée, pour construire une classe 
  de fonction \text{APN}. \\
  \\
  Posons : $B(x)=X^{2^{n/2}+1}$  et soit $G$ une $(n,n/2)-$fonction.\par
  \noindent Et définissons  $F :\,x\in \mathbb{F}_{2^{n}}\rightarrow (B(x),G(x)) \in \mathbb{F}_{2^{n/2}}\times \mathbb{F}_{2^{n/2}}. $\\
  \textbf{Probleme:} Donner une condition nécessaire et suffisante portant sur $G$ pour que $F$ soit \textbf{APN}.\\
  $F$ APN ssi $\forall a \in \mathbb{F^{*}}_{2^{n}},\forall (c,d) \in \mathbb{F}_{2^{n/2}}\times \mathbb{F}_{2^{n/2}}\; D_{a}F(X)=(c,d)$ à au plus $0$ ou deux solutions  dans $\mathbb{F}_{2^{n}}$.
  \begin{equation}\label{E:apnequ}
   \begin{cases}  
   B(X)+B(X+a)&=c \\
   G(X)+G(X+a)&=d  
    \end{cases}
   \end{equation}
   
  or $D_{a}B(X)=\trace(a^{2^{n/2}}X)+a^{2^{n/2}+1}$ et donc:\\
  $D_{a}B(X)=c \Leftrightarrow \trace(a^{2^{n/2}}X)+a^{2^{n/2}+1}=c \Leftrightarrow \trace(X)=1+\frac{c}{a^{2^{n/2}+1}} \quad\textrm{chgement de variable $X \to aX$ } $\\
  Soit $b\in \mathbb{F}_{2^{n}} \; \textrm{tel que} \, \trace(b)=1+\frac{c}{a^{2^{n/2}+1}}\; \textrm{ Voir la surjéctivité de la trace } $ et donc:\\
  $D_{a}B(X)=c \Leftrightarrow \trace(X+b)=0 \Leftrightarrow \trace(X)=0 \quad \textrm{changement de variable $X \to X+b$ }\Leftrightarrow  X \in \mathbb{F}_{2^{n/2}}$\\
  L'îsomorphisme affine est  $ \varphi :X \in \mathbb{F}_{2^{n/2}} \stackrel{\sim}{\rightarrow} aX+b \in (D_{a}B)^{-1}(c) \;\textrm{avec} \;\trace(b)=1+\frac{c}{a^{2^{n/2}+1}}  $\\
  En reportant dans l'equation~\eqref{E:apnequ} \\
  $F$ APN ssi 
  \begin{equation}
  \forall a \in \mathbb{F^{*}}_{2^{n}} ,\forall b \in \mathbb{F}_{2^{n}}, \forall d \in \mathbb{F}_{2^{n/2}}.
  G(aX+b)+G(aX+b+a)=d
  \end{equation} à au plus $0$ ou deux solutions  sur $\mathbb{F}_{2^{n/2}}$.\\
  On a le théorème suivant:
  \begin{thm}\label{first}
  Soit $B(x)=X^{2^{n/2}+1}$  et soit $G$ une $(n,n/2)-$fonction, $L : \mathbb{F}_{2^{n/2}}\times \mathbb{F}_{2^{n/2}}\rightarrow \mathbb{F}_{2^{n}}$ un isomorphisme linéaire et
  $ F : X \in \mathbb{F}_{2^{n}} \rightarrow L(B(x),G(x)) \in \mathbb{F}_{2^{n}}$ alors $F$ est APN si et seulement si
  \begin{equation}\label{important}
  \forall a \in \mathbb{F^{*}}_{2^{n}} ,\forall b \in \mathbb{F}_{2^{n}}, \forall d \in \mathbb{F}_{2^{n/2}} \;
  G(aX+b)+G(aX+b+a)=d
  \end{equation} à $0$ ou $2$ solutions au plus dans $\mathbb{F}_{2^{n/2}}$.\\
  \end{thm}
  
  \begin{lem} \label{gold}
  \indent
  \begin{enumerate}
  \item  $ \forall a,b \in \mathbb{F}_{2^{r}} , (aX+b)^{2^{k}+2^{j}}+(aX+a+b)^{2^{k}+2^{j}}=a^{2^{k}+2^{j}}(X^{2^{k}}+X^{2^j}+1)+b^{2^{k}}a^{2^{i}}+b^{2^{i}}a^{2^{k}}$
  \item Soient $i$ et $r$ premiers entre eux, $c$ élément de $\mathbb{F}_{2^{r}}$, alors l'équation : $ X^{2^{i}}+X+c=0$ à $0$ ou $2$ solutions au plus dans $\mathbb{F}_{2^{r}}$
  \end{enumerate}
  \end{lem}
  \begin{proof}
  Le point (i) est calculatoire, le point (ii) découle du fait que, $ X \rightarrow X^{2^{i}+1} $ est APN (\textbf{Gold}). 
  \end{proof}
  \subsubsection{Familles de fonctions connues}
     Je commencerais par trois lemmes que j'ai jugé fort utile.
     \begin{lem}\label{parité}
     Pour tout entier $i$
     \renewcommand{\labelenumi}{{\normalfont (\roman{enumi})}}
     \begin{enumerate}
     \item $2^{i}+1$ est divisible par $3$ ssi $i$ est impair.
     \item $2^{i}-1$ est divisible par $3$ ssi $i$ est pair.
     \end{enumerate}
     \end{lem}
     \begin{proof}
     Tout entier $i$ peut s'écrire : $i=2k+\epsilon \;\textrm{avec}\,(k,\epsilon)\in \mathbb{N}\times \left\{0,1\right\}$
     \begin{enumerate}
     \item $2^{i}+1=2^{2k+\epsilon}+1=4^{k}2^{\epsilon}+1\equiv (2^{\epsilon}+1)\mod3$ \\
        donc $2^{i}+1 \equiv 0 \mod3 \; \textrm{ssi}\; \epsilon =1$
      \item $2^{i}-1=2^{2k+\epsilon}-1=4^{k}2^{\epsilon}-1 \equiv (2^{\epsilon}-1)\mod3$ \\
      donc $2^{i}-1 \equiv 0 \mod3 \; \textrm{ssi}\; \epsilon =0$
     \end{enumerate}
     \end{proof}
     \begin{lem}\label{ssi}
     Soient $n$ un entier pair et $i$ un entier vérifiant $\gcd(i,n/2)=1$ alors 
     \begin{equation*}
   \gcd(2^{i}+1,2^{n/2}+1)=
   \begin{cases}
      1, &\text{ssi $i$ pair}\\
      1, &\text{ssi $i$ impair et $n/2$ pair}\\
      3, &\text{ssi $i$ impair et $n/2$ impair}
   \end{cases}
   \end{equation*}
   \end{lem}
   \begin{proof}
    
   On a d'un coté $\gcd(2^{2i}-1,2^{n}-1)=\gcd(2^{i}+1,2^{n}-1).\gcd(2^{i}-1,2^{n}-1)=\gcd(2^{i}+1,2^{n/2}-1)\gcd(2^{i}+1,2^{n/2}+1)(2^{\gcd(i,n)}-1)$.\\
   de l'autre $\gcd(2^{2i}-1,2^{n}-1)=2^{2.\gcd(i,n/2)}-1=3$.
   Soit \begin{equation*}
   \boxed{3=\gcd(2^{i}+1,2^{n/2}-1)\gcd(2^{i}+1,2^{n/2}+1)(2^{\gcd(i,n)}-1)}
   \end{equation*}
   on conclut on traitons selon la parité de $i$ et on utilisant le lemme~\ref{parité}.
   \end{proof}
   \begin{lem}\label{solut}
   
   Soit $q=2^{n/2}$, les solutions de l'équation 
   \begin{equation}\label{equ:sol}
   X^{q+1}+1=0
   \end{equation}
   sont exactement les éléments de $\mathbb{F}_{2^{n}}^{*(q-1)}$
   \end{lem}
   \begin{proof}
   En effet, soit $x\in \mathbb{F}_{2^{n}}^{*(q-1)}$, donc il existe $y\in \mathbb{F}_{2^{n}}^{*}$ vérifiant $x=y^{q-1}$, soit $x^{q+1}=y^{q^{2}-1}=1$ donc $x$ est solution de l'équation~\ref{equ:sol}. Or d'après la proposition~\ref{Per}, $\,\#{\mathbb{F}_{2^{n}}^{*(q-1)}}= q+1$.\\
   D'un autre coté les solutions de l'équation~\eqref{equ:sol} sont simples voir proposition \ref{double}, ils sont au nombre de $q+1$, vue que son degré est $q+1$. Ce qui achève la preuve.
   \end{proof}
  Dorénavant et dans toutes la suite, nous prendrons pas en compte, les termes de $\mathbb{F}_{2^{n/2}}$ indépendants de $X$ qui apparaissent dans $G(aX+b')+G(aX+b'+a)$, puisqu'on peut toujours les affectés à $d$ dans l'égalité~\eqref{important}.
  \begin{corol}\label{faux}
   La fonction 
    $F(X)=X^{2^{2i}+2^{i}}+bX^{q+1}+cX^{q(2^{2i}+2^{i})}$ où $\gcd(i,n/2)=1,\; q=2^{n/2}, \;(c,b )\in {\mathbb{F}_{2^{n}}}^{2},\;\textrm{tel que : }\;
     c^{q+1}=1,\; c \notin \left\{\lambda ^{(2^{i}+1)(q-1)}, \;\lambda \in \mathbb{F}_{2^{n}} \right\} \,\textrm{et }\; cb^{q}+b\neq 0$  est  \textbf{APN}
      \end{corol}
      \begin{proof}
      Nous allons commencer par quelques remarques simples:
\renewcommand{\labelenumi}{{\normalfont (\roman{enumi})}}
\begin{enumerate}
\item $c \notin \mathbb{F}_{2^{n/2}}$\\
En effet: Si $c\in \mathbb{F}_{2^{n/2}}$  alors $c^{q+1}=c^{2}=1 \Rightarrow c=1 \Rightarrow c \in \left\{\lambda ^{(2^{i}+1)(q-1)} , \lambda \in \mathbb{F}_{2^{n}} \right\}$, contradiction.
\item $\frac{b}{c^{2^{(n-1)}}} \notin \mathbb{F}_{2^{n/2}}$\\
  En effet, sinon : $\left(\frac{b}{c^{2^{(n-1)}}}\right)^{q}=\frac{b}{c^{2^{(n-1)}}}$ or $\left(\frac{1}{c^{2^{(n-1)}}}\right)^{q}=\frac{c}{c^{2^{(n-1)}}}$ (cf. proposition~\ref{clef})\\
  $\Leftrightarrow b^{q} \frac{c}{c^{2^{(n-1)}}}=\frac{b}{c^{2^{(n-1)}}} \Leftrightarrow cb^{q}+b=0$, contradiction.\\
  et donc $(1,\frac{b}{c^{2^{(n-1)}}})$ forme une base de $\mathbb{F}_{2^{n}}$ sur $\mathbb{F}_{2^{n/2}}$.
  \item $F$ est APN si et seulement si $\frac{F}{c^{2^{(n-1)}}}$ est APN ( evident ).
\end{enumerate}
Sans perte de généralité nous pouvons identifier $F$ à $\frac{F}{c^{2^{(n-1)}}}$ et donc \\
 $F= \frac{X^{2^{2i}+2^{i}}}{c^{2^{(n-1)}}}+\frac{b}{c^{2^{(n-1)}}} X^{q+1}+\frac{c}{c^{2^{(n-1)}}}X^{q(2^{2i}+2^{i})}=\trace\left( \frac{X^{2^{2i}+2^{i}}}{c^{2^{(n-1)}}}\right)+\frac{b}{c^{2^{(n-1)}}}X^{q+1}$\\
 
 Posons $L(u,v)=u+v\frac{b}{c^{2^{(n-1)}}}$ c'est bien un isomorphisme, et $ G(X)=\trace\left( \frac{X^{2^{2i}+2^{i}}}{c^{2^{(n-1)}}}\right).$
 
   Montrons que $G$ vérifie l'équation~(\ref{important}).
   \begin{align*}
   G(aX+b')+G(aX+a+b')&=\trace\left( \frac{a^{2^{2i}+2^{i}}}{c^{2^{(n-1)}}}(X^{2^{2i}}+X^{2^{i}}+1)\right)\\
                    &=(X^{2^{2i}}+X^{2^{i}}+1)\trace\left( \frac{a^{2^{2i}+2^{i}}}{c^{2^{(n-1)}}}\right)\\
                    &=(X^{2^{i}}+X+1)^{2^{i}} \trace\left( \frac{a^{2^{2i}+2^{i}}}{c^{2^{(n-1)}}}\right)\\
                    &=d
    \end{align*}
    Or $x \rightarrow x^{2^{i}}$ est une permutation sur $\mathbb{F}_{2^{n/2}}$, d'après le lemme~\ref{gold}, et le théorème~\ref{first}, $F$ est APN si et seulement si $ \forall a \in \mathbb{F^{*}}_{2^{n}} \quad      \trace\left( \frac{a^{2^{2i}+2^{i}}}{c^{2^{(n-1)}}}\right)=0 $ n'a pas de solution.\\Supposons que c'est le cas, ceci équivaut à 
   \begin{align*}
   \frac{a^{2^{2i}+2^{i}}}{c^{2^{(n-1)}}} &=\frac{a^{q(2^{2i}+2^{i})}}{c^{q 2^{(n-1)}}}\\
                                          &=\frac{c}{c^{2^{(n-1)}}} a^{q(2^{2i}+2^{i})}  \quad \textrm{( cf. proposition~\ref{clef} ) }
   \end{align*}
 
 Ceci implique $ c=a^{-(q-1)(2^{2i}+2^{i})}=a^{-(q-1)2^{i}(2^{i}+1)} \in \mathbb{F}_{2^{n}}^{(q-1)2^{i}(2^{i}+1)} $\\
 Comme $ x\rightarrow x^{2^{i}}$ est une permutation,
 ce qui entraîne  $c \in \mathbb{F}_{2^{n}}^{(q-1)(2^{i}+1)} $ contradiction, ce qui achève la preuve.
      \end{proof}
  \noindent  \textbf{Remarque:}
  \begin{enumerate}
  \item Les lemmes que j'ai donné \ref{parité}, \ref{solut} et surtout \ref{ssi}, sont très puissantes, ils trouveront leur application dans ce qui va suivre, mais peuvent être appliqués en dehors de l'article.
\item 
 L'étude faite à la sous-section~\ref{per}, montre que $\;\mathbb{F}_{2^{n}}^{(q-1)(2^{i}+1)}=\mathbb{F}_{2^{n}}^{\gcd((q-1)(2^{i}+1),2^{n}-1)}=\mathbb{F}_{2^{n}}^{(q-1)\gcd(2^{i}+1,q+1)}$\\
On utilisant le lemme~\ref{ssi}, on a:
\begin{equation*}
  \mathbb{F}_{2^{n}}^{(q-1)(2^{i}+1)} =
   \begin{cases}
      \mathbb{F}_{2^{n}}^{(q-1)}, &\text{ssi $i$ pair}\\
      \mathbb{F}_{2^{n}}^{(q-1)}, &\text{ssi $i$ impair et $n/2$ pair}\\
      \mathbb{F}_{2^{n}}^{3(q-1)}, &\text{ssi $i$ impair et $n/2$ impair}
   \end{cases}
   \end{equation*}
   \item Le corollaire~\ref{faux}, n'est pas tout à fait correcte, car ils y'a des cas où les hypotheses ne seront jamais satisfaites, et ça c'est très important quand on implémente, de chercher la où on peut trouver. En effet le lemme~\ref{solut}, supprime les deux cas : $i$ pair et $i$ impair avec $n/2$ pair, nous allons donné une version corrigée est optimale de ce résultat.
  \end{enumerate} 
  
  \begin{corol}[version optimale]
  
  Soient $q=2^{n/2}$, $i$ et $n/2$ impairs, vérifiant $\gcd(i,n/2)=1$. Alors:\\
  La fonction 
    $F(X)=X^{2^{2i}+2^{i}}+bX^{q+1}+cX^{q(2^{2i}+2^{i})}$ où   $ \;c,b \in \mathbb{F}_{2^{n}},\textrm{ tel que : }\;
     c^{q+1}=1,\, c \notin \mathbb{F}_{2^{n}}^{*3(q-1)}, \textrm{et }\; c b^{q}+b\neq 0$  est  \textbf{APN}
      \end{corol}

\begin{corol}
Soient $q=2^{n/2}$, $s$ et $n/2$ impairs, vérifiant $\gcd(s,n/2)=1$,
 $b\in \mathbb{F}_{2^{n}}$ non cube, et $c\in \mathbb{F}_{2^{n}}\setminus \mathbb{F}_{2^{n/2}} $, $r_{i} \in \mathbb{F}_{2^{n/2}}$. Alors la fonction $F$ définit par:\\
   $F(X)=b X^{2^{s}+1}+b^{q} X^{q(2^{s}+1)}+c X^{q+1}+\sum_{i=1}^{n/2-1} r_{i} X^{2^{i}(q+1)}$ est APN.

\end{corol}
\begin{proof}

On a l'isomorphisme suivant $L(u,v)=cu+\sum_{i=1}^{n/2-1} r_{i} u^{2^{i}}+v \; (\textrm{cf :c $\in \mathbb{F}_{2^{n}}\setminus \mathbb{F}_{2^{n/2}}$}) $\\
et donc $G(X)= b X^{2^{s}+1}+b^{q} X^{q(2^{s}+1)}=\trace\left(bX^{2^{s}+1}\right)$\\
Montrons que $G$ vérifie l'équation~(\ref{important})

\begin{align*}
   G(aX+b')+G(aX+a+b')&=\trace\left(b a^{2^{s}+1}(X^{2^{s}}+X+1) \right)\\
                    &=(X^{2^{s}}+X+1)\trace\left(b a^{2^{s}+1} \right)\\
                    &=d
    \end{align*}
  D'après le lemme~\ref{gold}, et le théorème~\ref{first}, $F$ est APN si et seulement si  $ \forall a \in \mathbb{F^{*}}_{2^{n}} \quad \trace\left(b a^{2^{s}+1} \right)=0$ n'a pas de solutions.
  C'est le cas, sinon : $\;b a^{2^{s}+1}\in \mathbb{F}_{2^{n/2}}$ comme $\mathbb{F}_{2^{n}}^{(2^{s}+1)}= \mathbb{F}_{2^{n}}^{\gcd(2^{s}+1,2^{n}-1)}=\mathbb{F}_{2^{n}}^{3}$
  
  et que les élément de $\mathbb{F}_{2^{n/2}}$ sont tous des cubes, ceci conduit à  $b$ est un cube, contradiction.
  Donc $F$ est APN.
\end{proof}
  
  \begin{corol} \label{gener}
Soient $\gcd(i,n/2)=1 \; ,\; c\in \mathbb{F}_{2^{n}} \;,\; s \in \mathbb{F}_{2^{n}}\setminus \mathbb{F}_{2^{n/2}} \;,\; q=2^{n/2}$\\
  $F(X)=X(X^{2^{i}}+X^{q}+c X^{2^{i}q})+X^{2^{i}}(c^{q}X^{q}+s X^{q2^{i}})+X^{(2^{i}+1)q}$.\\
  où : $X^{2^{i}+1}+c X^{2^{i}}+c^{q}X+1$ est irréductible sur $\mathbb{F}_{2^{n}}$. Alors $F$ est APN.
  
\end{corol}

\begin{proof}
  \begin{align*}
     F(X)&=X^{2^{i}+1}+X^{q+1}+c X^{2^{i}q+1}+c^{q}X^{q+2^{i}}+s X^{2^{i}(q+1)}+X^{(2^{i}+1)q}.\\
         &=X^{q+1}+s X^{2^{i}(q+1)}+X^{2^{i}+1}+c X^{2^{i}q+1}+c^{q} X^{q+2^{i}}+X^{(2^{i}+1)q}\\
         &=L(B(X),G(X))
 \end{align*}
 où : $L(u,v)= u+su^{2^{i}}+v$ est un isomorphisme; $G(X)=\trace\left(X^{2^{i}+1}+cX^{2^{i}q+1}\right)$
 Montrons que $G$ vérifie l'équation~\eqref{important}
 \begin{align*}
   G(aX+b')+G(aX+a+b')&=\trace\left(a^{2^{i}+1}(X^{2^{i}}+X+1)+ca^{2^{i}q+1}(X^{2^{i}q}+X+1)\right)  \quad (\textrm{ comme $\,X\in \mathbb{F}_{2^{n/2}}$})\\
                      &=(X^{2^{i}}+X+1) \trace\left(a^{2^{i}+1}+ca^{2^{i}q+1}\right)\\
                      &=d.
   \end{align*}
   Supposons que $\trace\left(a^{2^{i}+1}+ca^{2^{i}q+1}\right)=0$ a une solution dans $\mathbb{F}_{2^{n/2}}$, ceci équivaut a : 
       
  $a^{2^{i}+1}+ca^{2^{i}q+1}=a^{q(2^{i}+1)}+c^{q}a^{2^{i}+q}$, en divisant par $a^{2^{i}+1}$, et en notant :\\$P(X)=X^{2^{i}+1}+c X^{2^{i}}+c^{q}X+1$, ceci donne $P(a^{q-1})=0$, contradiction.
\end{proof}
On voit que le théorème~\ref{first}, permet de construire une multitude de fonction APN, moi même j'en donne deux, assez générales.\\
posons $G(X)=\trace\left(X^{2^{i}+1}+cX^{2^{i}q+1}+tX^{2^{i}+q}\right)$ et trouvons les conditions nécéssaire et suffisante pour que ca conduit a une fonction APN.
\begin{align*}
   G(aX+b')+G(aX+a+b')&=\trace\left(a^{2^{i}+1}(X^{2^{i}}+X+1)+ca^{2^{i}q+1}(X^{2^{i}q}+X+1)+ta^{q+2^{i}}(X^{2^{i}}+X+1)\right)  \quad (\textrm{ comme $\,X\in \mathbb{F}_{2^{n/2}}$})\\
                      &=(X^{2^{i}}+X+1) \trace\left(a^{2^{i}+1}+ca^{2^{i}q+1}+ta^{q+2^{i}}\right)\\
                      &=d.
   \end{align*}
  Supposons que $\trace\left(a^{2^{i}+1}+ca^{2^{i}q+1}+ta^{q+2^{i}}\right)=0$ a une solution dans $\mathbb{F}_{2^{n}}$. Ceci équivaut a : \\
  $P(a^{q-1})=0$ où $P(X)=X^{2^{i}+1}+(t^{q}+c)X^{2^{i}}+(c^{q}+t)X+1$, il suffit de choisir $P$ irréductible sur $\mathbb{F}_{2^{n}}$
  \begin{corol}
  Soient $q=2^{n/2}$, $i$ tel que $\gcd(i,n/2)=1$, $B(X)=X^{q+1}$.\\$G(X)=\trace\left(X^{2^{i}+1}+cX^{2^{i}q+1}+tX^{2^{i}+q}\right)$, et $L : \mathbb{F}_{2^{n/2}}\times \mathbb{F}_{2^{n/2}}\rightarrow \mathbb{F}_{2^{n}}$ un isomorphisme quelconque.\\
  Alors   $F=L(B,G)$  est \textbf{APN} si et seulement si $P(X)=X^{2^{i}+1}+(t^{q}+c)X^{2^{i}}+(c^{q}+t)X+1$ n'a pas de racines dans $\mathbb{F}_{2^{n}}$.
  \end{corol}
  
\begin{corol}
Soient $n/2$ impair, et $i,j$ vérifiant $(j-i)$ impairs et $\gcd(j-i,n/2)=1$, $c$ élément de $ \mathbb{F}_{2^{n}}^{*(q-1)} \setminus \mathbb{F}_{2^{n}}^{*3(q-1)}$, $B(X)=X^{q+1}$.
$G(X)=\trace\left(\frac{X^{2^{j}+2^{i}}}{c^{2^{n-1}}}\right)$, et $L : \mathbb{F}_{2^{n/2}}\times \mathbb{F}_{2^{n/2}}\rightarrow \mathbb{F}_{2^{n}}$ un isomorphisme quelconque. Alors\\
$F=L(B,G)$  est \textbf{APN}.
\end{corol}
 \begin{proof}
 Vérifiant que $G$ satisfait l'équation~\eqref{important}.
 \begin{align*}
   G(aX+b')+G(aX+a+b')&=\trace\left(\frac{a^{2^{j}+2^{i}}}{c^{2^{n-1}}}(X^{2^{j}}+X^{2^{i}}+1)\right)\\
                      &=(X^{2^{j-i}}+X+1)^{2^{i}}\trace\left(\frac{a^{2^{j}+2^{i}}}{c^{2^{n-1}}}\right)\\
                      &=d
  \end{align*}
  Les mêmes arguments utiliser jusqu'ici, montre que $F$ est APN ssi $\trace\left(\frac{a^{2^{j}+2^{i}}}{c^{2^{n-1}}}\right)=0$ n'a pas de solutions pour tout $a$ dans ${\mathbb{F}^{*}}_{2^{n}}$. Supposons que c'est le cas alors $c \in \mathbb{F}_{2^{n}}^{*(q-1)(2^{j}+2^{i})}$ soit d'apres le lemme~\ref{ssi}, $c\in\mathbb{F}_{2^{n}}^{*3(q-1)}$ contradiction, donc $F$ est APN.
 
 \end{proof}

\end{document}